\documentclass[12pt,a4paper]{article}

\usepackage[a4paper,margin=1in]{geometry}
\usepackage{threeparttable}

\usepackage{setspace}
\usepackage[T1]{fontenc}
\usepackage{lmodern}
\usepackage{amsmath,amssymb,amsthm}
\usepackage{latexsym}
\usepackage{bm}
\usepackage{xcolor}
\usepackage[round,authoryear]{natbib}
\usepackage{booktabs}
\usepackage{tabularx}
\usepackage{graphicx}
\usepackage{float}
\usepackage{hyperref}
\hypersetup{
  colorlinks=true,
  linkcolor=blue,
  citecolor=blue,
  urlcolor=blue
}

\newtheorem{theorem}{Theorem}
\newtheorem{proposition}{Proposition}
\newtheorem{corollary}{Corollary}
\newtheorem{lemma}{Lemma}

\theoremstyle{definition}

\theoremstyle{remark}
\newtheorem{remark}{Remark}

\newcommand{\PP}{\mathbb{P}}
\newcommand{\dnormal}{\mathcal{N}}
\newcommand{\dgamma}{\mathcal{G}}

\title{Exact two-stage finite-mixture representations for species sampling processes}

\author{
Rams\'es H. Mena \and
Christos Merkatas \and
Theodoros Nicoleris \and
Carlos E. Rodr\'iguez
}

\date{} 

\begin{document}

\maketitle

\begin{abstract}
Discrete random probability measures are central to Bayesian inference, particularly as priors for mixture modeling and clustering. A broad and unifying class is that of proper species sampling processes (SSPs), encompassing many Bayesian nonparametric priors.   We show that any proper SSP admits an exact two-stage finite-mixture representation built from a latent truncation index and a simple reweighting of the atoms. For each realized truncation index, the representation has finitely many atoms, and averaging over the induced law of that index recovers the original SSP setwise. This yields at least two consequences: (i) an exact two-stage finite construction for arbitrary SSPs, without user-chosen truncation levels; and (ii) posterior inference in SSP mixture models via standard finite-mixture machinery, leading to tractable MCMC algorithms without ad hoc truncations. We explore these consequences by deriving explicit total-variation bounds for the approximation error when the truncation level is fixed, and by studying practical performance in mixture modeling, with emphasis on Dirichlet and geometric SSPs.
\end{abstract}

\noindent{\small\textbf{Keywords:}
Bayesian nonparametrics; clustering; mixture model; posterior computation; random probability measure.

\section{Introduction}
Random probability measures, together with their constructions, representations and associated algorithms, play a central role in Bayesian nonparametrics and in the design of flexible statistical models whose posterior computation is often analytically intractable. The canonical example is the Dirichlet process~\citep{Ferguson1973}, which provides a flexible prior on probability measures while retaining analytic and computational tractability. Over the years, a variety of alternative priors have been proposed, including normalized completely random measures \citep{regazzini2003}, Gibbs–type priors \citep{GnedinPitman2006,LijoiMenaPrunster2007Biometrika,DeBlasi2015gibbstype}
and many stick–breaking constructions \citep{sethuraman1994constructive, PitmanYor1997, ishwaran2001, Gil-Leyva2023}. A unifying perspective on this landscape is provided by \emph{species sampling processes} (SSPs), which capture a broad class of discrete random probability measures and the exchangeable sequences they induce.

An SSP on a Polish space $(\mathbb{X}, \mathcal{B}_{\mathbb{X}})$ is a random probability measure of the form
\begin{equation}
G(A) = \sum_{j=1}^\infty w_j\,\delta_{\theta_j}(A)
+
\Bigg(1 - \sum_{j=1}^\infty w_j\Bigg)G_0(A),
\label{eq:ssp-general}
\end{equation}
with $A\in\mathcal{B}_{\mathbb{X}}$ and where the atoms $(\theta_j)_{j\geq 1}$, are independent and identically distributed (i.i.d.) from a diffuse base measure $G_0$ on $(\mathbb{X},\mathcal{B}_{\mathbb{X}})$, the weights $w_j\ge0$ satisfy $\sum_j w_j\le 1$ almost surely (a.s.), and $(\theta_j)_j$ is independent of $(w_j)_j$~\citep{Pitman1996ssm,pitman2006}. When $\sum_j w_j = 1$ a.s., the diffuse component vanishes and $G$ becomes an a.s.\  \emph{proper} SSP. From a Bayesian viewpoint, proper SSPs are the relevant objects, as a fixed diffuse component cannot learn from the data. Throughout we focus on proper SSPs.

By de Finetti’s theorem, any exchangeable sequence $(X_i)_{i\ge1}$ with $\mathbb{X}$-valued elements is conditionally i.i.d., given a random directing measure $G$, and a large class of such sequences arises by taking $G$ to be a proper SSP. In this case we speak of a \emph{species sampling model} (SSM) for $(X_i)_{i\ge1}$ driven by the SSP $G$. The induced clustering structure can be described in terms of the associated exchangeable random partition and its exchangeable partition probability function (EPPF), or, when available, via predictive distributions generalizing the Blackwell–MacQueen Pólya urn~\citep{blackwell1973ferguson,pitman2006}. The connection between SSPs, partitions and predictive rules has made them a backbone of Bayesian nonparametric analysis; see, for example, \citet{LijoiPrunster2010}, \citet{LeeQuintanaMuellerTrippa2013} and \cite{ghosal2017}.

Several constructions have proved  influential in specifying the law of some SSPs. One route proceeds via prediction rules and EPPFs \citep{James2009}, and another via Gibbs–type priors and related classes~\citep{GnedinPitman2006, DeBlasi2015gibbstype}. Central to this work is the \emph{stick–breaking} representation of the weights 
\begin{equation}
w_1 = v_1,\qquad
w_j = v_j \prod_{i=1}^{j-1} (1-v_i),\quad j\ge2,
\label{eq:sb}
\end{equation}
where ``length'' variables $v_j\in(0,1)$ may be independent or dependent. This representation,  initially presented for the Dirichlet process \citep{sethuraman1994constructive},  has been extended in many directions; see, e.g.,~\citet{ishwaran2001,Dunson2008, Favaro2012,Favaro2016,Gil-Leyva2026}. A key point is that any SSP can be represented via stick–breaking weights~\citep{pitman2006,Gil-Leyva2023}.

A useful application of SSPs is in mixture settings. Given a kernel $f(x\mid\theta)$ on $\mathbb{X}$,
an SSP on the parameter space induces the random density~\citep{lo1984}
\begin{equation}
    X \mid G \sim f_G(x)
=\int f(x\mid\theta)\,G(\mathrm{d}\theta)
=\sum_{j=1}^\infty w_j\, f(x\mid\theta_j),
\label{eq:ssp-mixture}
\end{equation}
where $G = \sum_{j\ge1} w_j \delta_{\theta_j}$ almost surely.

This contrasts with classical finite mixtures, which instead posit a fixed number of components $m$ and
model the sampling density as a finite sum of $m$ kernels (with weights $w_{1:m}$ and atoms
$\theta_{1:m}$). In an SSP mixture, by contrast, the discreteness of $G$ induces a random partition of
a sample of size $n$ and thus a posterior distribution for the number of occupied clusters, say $c_n$,
which is a data-dependent occupancy statistic  rather than a
structural model parameter like $m$.

Familiar Bayesian nonparametric mixtures arise as special cases of~\eqref{eq:ssp-mixture},
including mixtures of Dirichlet, Pitman--Yor process  and a wide range of  models
studied in the literature; see, for instance,
\citet{FuentesGarcia2010, LijoiPrunster2010, Gil-Leyva2020,Gil-Leyva2023}.

While SSPs provide a rich modeling framework, their practical usefulness depends on
representations that yield computationally efficient and numerically stable methods. To this end,
SSPs admit several formulations, including EPPFs and predictive rules \citep{LijoiPrunster2010},
stick--breaking weights, and latent--variable slice-sampling constructions \citep{Walker2007,Kalli2011} or retrospective sampling \citep{pastp2008},
which have enabled a variety of applications and extensions \citep[see, e.g.,][]{Ni2020anchor,Canale2022,DeBlasiMafer2023}.
However, for many SSPs of practical interest, neither the EPPF nor the predictive distribution is
available in closed form, and generic stick--breaking samplers often rely on random truncation levels
that are hard to control and can be inefficient or unstable, especially when the number of active
atoms grows rapidly across iterations.

We show that any proper SSP admits an exact two-stage finite-mixture representation built from a latent truncation variable $K$ and reweighting the original atoms. Concretely, given a proper SSP with weights $\boldsymbol{w}=(w_j)_{j\ge1}$ and atoms $\boldsymbol{\theta}=(\theta_j)_{j\ge1}$,
\[
G(A\mid \boldsymbol{w},\boldsymbol{\theta})=\sum_{j=1}^\infty w_j\delta_{\theta_j}(A),
\]
there exists a random pair $(K,\tilde{\boldsymbol{w}})$ such that 
\[
G^\star(A\mid K,\tilde{\boldsymbol{w}},\boldsymbol{\theta})=\sum_{j=1}^K \tilde w_j\delta_{\theta_j}(A),
\] has a fully specified law, and averaging over $K$ recovers the original SSP setwise.

This latent representation has immediate consequences. It yields an exact two-stage  finite construction for any proper SSP, in the spirit of \cite{FergusonKlass1972},
\citet{ishwaran2002} and \citet{ArbelDeBlasiPrunster2019}. Rather than fixing a global truncation level and controlling an approximation error, one samples the truncation index $K$ from its prescribed law and then generates the associated  finite random measure; averaging over $K$ recovers exactly  the original infinite expansion at the level of set masses. It also enables standard finite-mixture machinery (allocations and Gibbs updates) for arbitrary SSPs without ad hoc cutoffs. More broadly, it separates representational convenience from modeling assumptions by disentangling the auxiliary $K$, the data-driven occupancy $c_n$, and the structural component count in genuinely finite mixture models.

\section{From infinite to finite}

Building on Section~1, we now make the  representation explicit for a generic proper SSP.
We introduce an auxiliary truncation index $K$ and reweighted weights $\tilde{\boldsymbol w}$, give their  specified law, and prove that averaging the resulting 
finite random measure over $K$ recovers the original SSP setwise.

\begin{theorem}\label{thm:finite-ssrm}
Let $G$ be a proper SSP on $(\mathbb{X},\mathcal{B}_{\mathbb{X}})$ admitting the a.s. representation
\begin{equation}
G(A \mid \boldsymbol{w}, \boldsymbol{\theta})
=
\sum_{j=1}^\infty w_j \, \delta_{\theta_j}(A),
\qquad A\in\mathcal{B}_{\mathbb{X}},
\label{pSSM}
\end{equation}
where $\boldsymbol{w}=(w_j)_{j\ge1}$ satisfies $w_j\ge0$ and $\sum_{j=1}^\infty w_j=1$ a.s., and
$\boldsymbol{\theta}=(\theta_j)_{j\ge1}$ are independent and identically distributed draws from $G_0$, a diffuse measure on
$(\mathbb{X},\mathcal{B}_{\mathbb{X}})$. Let $\boldsymbol{\xi}:=(\xi_j)_{j\ge1}$ be an a.s. strictly decreasing sequence in $(0,1]$ such that $\xi_j\to0$ as $j\to\infty$.
For each realization of $(\boldsymbol{w},\boldsymbol{\xi})$, define a  random variable $K$ on $\mathbb{N}$ with conditional distribution
\[
\mathbb{P}(K=k\mid\boldsymbol{w},\boldsymbol{\xi})
=
(\xi_k-\xi_{k+1})\, s_k,
\qquad
s_k:=\sum_{h=1}^k \xi_h^{-1} w_h,
\quad k\ge1.
\]
Then $\mathbb{P}(K=\cdot\mid\boldsymbol{w},\boldsymbol{\xi})$ is a well-defined probability mass function (almost surely).

Conditionally on $(\boldsymbol{w},\boldsymbol{\theta},\boldsymbol{\xi},K=k)$, define the finite random measure
\begin{equation}
G^\star(A\mid K=k,\tilde{\boldsymbol{w}},\boldsymbol{\theta},\boldsymbol{\xi})
=
\sum_{j=1}^k \tilde w_j\,\delta_{\theta_j}(A),
\qquad
\tilde w_j:=\frac{\xi_j^{-1}w_j}{s_k}.
\label{finiterep}
\end{equation}
Then, for every $A\in\mathcal{B}_{\mathbb{X}}$,
\[
\sum_{k=1}^\infty
G^\star(A \mid \tilde{\boldsymbol{w}}, \boldsymbol{\theta},\boldsymbol{\xi}, K = k)
\, \mathbb{P}(K = k \mid \boldsymbol{w},\boldsymbol{\xi})
=
G(A \mid \boldsymbol{w},\boldsymbol{\theta})
\qquad\text{a.s.}
\]
\end{theorem}
Equivalently, the theorem states that the original SSP is recovered exactly setwise after averaging with respect to the law of $K$ given $(\boldsymbol{w},\boldsymbol{\xi})$. It does not claim that, for a fixed realized value $K=k$, the 
finite measure $G^\star(\cdot \mid K=k,\tilde{\boldsymbol{w}},\boldsymbol{\theta},\boldsymbol{\xi})$ has the same law as the original SSP.
\begin{proof}
Let
$
\Omega_\xi:=\big\{\xi_1>\xi_2>\cdots,\ \xi_j\in(0,1],\ \forall\ j,\ \text{and }\xi_j\to0\big\}.
$ 
By assumption, $\mathbb{P}(\Omega_\xi)=1$. We work on $\Omega_\xi$ and condition on a fixed realization of
$(\boldsymbol{w},\boldsymbol{\xi})$.
We first check that $\mathbb{P}(K = k \mid \boldsymbol{w},\boldsymbol{\xi})$ defines a valid pmf.
Since $\{\xi_k\}$ is strictly decreasing in $(0,1]$ on $\Omega_\xi$, $\xi_k-\xi_{k+1}>0$ for all $k$, hence
$\mathbb{P}(K = k \mid \boldsymbol{w},\boldsymbol{\xi}) \geq 0$. Moreover, using Tonelli’s theorem,

\begin{align}
\sum_{k=1}^\infty \mathbb{P}(K = k \mid \boldsymbol{w},\boldsymbol{\xi})
&= \sum_{k=1}^\infty (\xi_k - \xi_{k+1}) \, s_k \nonumber\\
&= \sum_{j=1}^\infty \sum_{k=j}^\infty (\xi_k - \xi_{k+1}) \, \xi_j^{-1} w_j,
\label{eq:change-order-rand}\\
&= \sum_{j=1}^\infty \xi_j^{-1} w_j \sum_{k=j}^\infty (\xi_k - \xi_{k+1}),
\label{eq:factor-out-rand}\\
&= \sum_{j=1}^\infty \xi_j^{-1} w_j \,
\big(\lim_{l\to\infty}(\xi_j-\xi_{l+1})\big) \nonumber\\
&= \sum_{j=1}^\infty \xi_j^{-1} w_j \, \xi_j
= \sum_{j=1}^\infty w_j
=1,
\label{eq:telescoping-rand}
\end{align}
where \eqref{eq:telescoping-rand} uses $\xi_l\to0$ on $\Omega_\xi$. This proves that
$\mathbb{P}(K=\cdot\mid\boldsymbol{w},\boldsymbol{\xi})$ is a proper pmf.  Next, we verify that marginalizing over $K$ recovers the original  measure. For any $A\in\mathcal{B}_{\mathbb{X}}$,
\begin{align}
\sum_{k=1}^\infty
G^\star(A \mid \tilde{\boldsymbol{w}}, \boldsymbol{\theta},\boldsymbol{\xi}, K = k)
\, \mathbb{P}(K = k \mid \boldsymbol{w},\boldsymbol{\xi})
&= \sum_{k=1}^\infty
\left[ \sum_{j=1}^k \tilde{w}_j \, \delta_{\theta_j}(A) \right]
(\xi_k - \xi_{k+1}) s_k \nonumber\\
&= \sum_{k=1}^\infty \sum_{j=1}^k
\frac{\xi_j^{-1} w_j}{s_k} \, \delta_{\theta_j}(A) \,
(\xi_k - \xi_{k+1}) s_k \nonumber\\
&= \sum_{k=1}^\infty \sum_{j=1}^k
(\xi_k - \xi_{k+1}) \, \xi_j^{-1} w_j \, \delta_{\theta_j}(A) \nonumber\\
&= \sum_{j=1}^\infty \sum_{k=j}^\infty
(\xi_k - \xi_{k+1}) \, \xi_j^{-1} w_j \, \delta_{\theta_j}(A) \label{cslice}\\
&= \sum_{j=1}^\infty \xi_j^{-1} w_j \, \delta_{\theta_j}(A)
\sum_{k=j}^\infty(\xi_k-\xi_{k+1}) \nonumber\\
&= \sum_{j=1}^\infty \xi_j^{-1} w_j \, \delta_{\theta_j}(A)\,\xi_j \nonumber\\
&= \sum_{j=1}^\infty w_j\,\delta_{\theta_j}(A)
= G(A \mid \boldsymbol{w},\boldsymbol{\theta}),
\end{align}
which completes the proof.
\end{proof}

Theorem~\ref{thm:finite-ssrm}  shows that any proper SSP admits an exact two-stage  finite representation with random truncation level $K$, whose setwise average recovers the original SSP. The inspiration is taken from the 
method by \cite{Kalli2011}. Introduce an auxiliary variable $u\in(0,1)$ and define
\[
G(A,u \mid \boldsymbol{w},\boldsymbol{\theta})
=
\sum_{j=1}^{\infty} \xi_j^{-1}\,\mathbb{I}(u\leq \xi_j)\, w_j \,\delta_{\theta_j}(A),
\]
which defines a joint kernel on $\mathcal{B}_{\mathbb{X}}\times(0,1)$ so that
\begin{align*}
\int_{0}^1  G(A,u \mid \boldsymbol{w}, \boldsymbol{\theta})\,\mathrm{d}u 
&= \int_0^1 \sum_{j=1}^{\infty}
\xi_j^{-1}\mathbb{I}(u\leq \xi_j) w_j \delta_{\theta_j}(A)\,\mathrm{d}u, \\
&= \sum_{k=1}^{\infty}\int_{\xi_{k+1}}^{\xi_k}
\sum_{j=1}^{\infty}
\xi_j^{-1}\mathbb{I}(u\leq \xi_j) w_j \delta_{\theta_j}(A)\,\mathrm{d}u, \\
&= \sum_{k=1}^{\infty}\int_{\xi_{k+1}}^{\xi_k}
\sum_{j=1}^{k} \xi_j^{-1} w_j \delta_{\theta_j}(A)\,\mathrm{d}u, \\
&= \sum_{k=1}^{\infty} \sum_{j=1}^k
(\xi_k - \xi_{k+1}) \xi_j^{-1} w_{j}\delta_{\theta_j}(A).
\end{align*}
Thus we recover (\ref{cslice}), 
which serves as the bridge between the  finite representation and the original SSP obtained setwise after averaging over $K$. In the slice-sampling formulation, the latent truncation level $K$ corresponds to the index of the last atom ``visible'' for a given slice
$u$, and the conditional law $\mathbb{P}(K=k\mid\boldsymbol{w}, \boldsymbol{\xi})$ arises
 from the lengths of the intervals $(\xi_{k+1},\xi_k]$ and the
weights $w_1,\dots,w_k$.

For a fixed SSP specified through its infinite weight sequence $\{w_j\}_{j\ge1}$, the theorem yields a family of exact  finite representations indexed by a decreasing sequence $\{\xi_j\}\downarrow0$. Different choices of $\{\xi_j\}$ (deterministic or random) simply induce different conditional laws for the latent truncation level $K$ given $w$. Such representations can be exploited for computational efficiency or theoretical explorations, without altering the underlying SSP at the setwise level after averaging over $K$.

An interesting case arises when we  consider a \emph{natural} choice of random $\boldsymbol{\xi}$ derived via the tail mass of the stick-breaking construction. 

\begin{corollary}\label{C2}
Let $G$ be a species sampling process with stick--breaking representation
\[
w_1 = v_1,\qquad
w_j = v_j \prod_{l=1}^{j-1} (1 - v_l), \quad j\ge2,
\]
where $v_j\in(0,1)$ almost surely and the resulting weights satisfy
$\sum_{j\ge1} w_j = 1$ almost surely. Define
\[
\xi_j = \prod_{l=1}^{j-1} (1 - v_l),\qquad \xi_1 = 1.
\]
Then $\{\xi_j\}$ is a.s. decreasing in $(0,1]$ with $\xi_j\downarrow 0$,
and Theorem~\ref{thm:finite-ssrm} applies with
\[
s_k = \sum_{j=1}^k \xi_j^{-1}w_j = \sum_{j=1}^k v_j,
\qquad
\tilde w_j = \frac{\xi_j^{-1}w_j}{s_k} = \frac{v_j}{s_k}.
\]
Consequently, conditional on $(\boldsymbol v,\boldsymbol\theta,K=k)$,
\[
G^\star(A \mid \boldsymbol v,\boldsymbol\theta,K=k)
= \frac{1}{s_k}\sum_{j=1}^k v_j\,\delta_{\theta_j}(A),
\qquad A\in\mathcal{B}_{\mathbb{X}},
\]
and the truncation level has conditional distribution
\[
\mathbb P(K=k\mid \boldsymbol v)
= (\xi_k-\xi_{k+1})s_k
= w_k\, \sum_{j=1}^k v_j.
\] 
\end{corollary}

Dirichlet processes ($v_j\overset{iid}\sim\mathrm{Beta}(1,\alpha)$), two-parameter Pitman--Yor processes
($v_j\sim\mathrm{Beta}(1-\sigma,\alpha+j\sigma)$), and related stick--breaking priors fit directly into
Corollary~\ref{C2}. A particularly transparent special case is the geometric stick--breaking process:
if $v_j\equiv v$ (with $v\sim\mathrm{Beta}(a,b)$), then $w_j=v(1-v)^{j-1}$ and $\xi_j=(1-v)^{j-1}$, so that
$s_k=\sum_{j=1}^k v_j=k\,v$ and $\tilde w_j=v_j/s_k=1/k$. Hence, conditional on $K=k$,
\[
G^\star(A\mid \boldsymbol\theta,K=k)=\frac{1}{k}\sum_{j=1}^k \delta_{\theta_j}(A),
\qquad A\in\mathcal{B}_{\mathbb{X}},
\]
and the truncation level simplifies to
\[
\mathbb P(K=k\mid v)=w_k\sum_{j=1}^k v_j
= k\,v^2(1-v)^{k-1},\qquad k=1,2,\ldots,
\]
which is a proper pmf since $\sum_{k\ge1}k\,v^2(1-v)^{k-1}=1$.

    Thus, in the geometric stick-breaking case, Theorem~\ref{thm:finite-ssrm} yields a particularly simple  finite representation with equal weights $1/k$ on the first $k$ atoms; after averaging over $K$, this two-stage construction recovers the original geometric stick-breaking SSP of \cite{FuentesGarcia2010} in the setwise sense of Theorem~\ref{thm:finite-ssrm}.

Our representation has two immediate implications. First, it provides an exact two-stage  finite construction for any proper SSP, where the target process is recovered setwise after averaging over the auxiliary truncation variable $K$. Second, it yields a natural finite-dimensional augmentation that enables posterior computation within standard Bayesian nonparametric methods. The next two sections develop these two uses.

\section{Simulation via the two-stage representation}
The finite mixture representation in  Theorem~\ref{thm:finite-ssrm} yields a direct simulation mechanism for any proper SSP prior. By introducing a latent truncation level $K$ and reweighting the first $K$ atoms, the construction produces a finite random measure $G^\star(\cdot \mid K,\tilde w,\theta,\xi)$. Averaging over $K$ yields a measure that coincides with the original SSP on every measurable set. In this sense, the construction provides an exact two-stage finite construction for prior simulation.
Unlike classical truncation schemes, the truncation index $K$ is not chosen to control an approximation error. Rather, $K$ is sampled from its induced law and then the corresponding  finite measure is generated; averaging over $K$ recovers the target SSP for any measurable set. This stands in contrast to deterministic truncations, where the cutoff is fixed in advance to trade accuracy for computational cost.

To place this construction in context, consider the approach by
\cite{ArbelDeBlasiPrunster2019}, who propose a simulation scheme for the
Pitman--Yor process (PYP) that achieves \emph{almost sure} error control, with particular simplifications in the Dirichlet process (DP) case corresponding to discount parameter $\sigma=0$.

As before, let $G=\sum_{j\ge1} w_j \delta_{\theta_j}$, with $\sum_{j\ge1} w_j = 1$ a.s. and denote the remainder (tail) mass as $R_n := 1 - \sum_{j=1}^n w_j = \sum_{j>n} w_j$. For $\varepsilon \in (0,1)$, introduce the stopping time
$\label{eq:tau_eps_directsim}
\tau(\varepsilon) := \min\{ n \ge 1 : R_n < \varepsilon \}.
$ 
A finite approximation is then obtained by \emph{lumping the tail}
into a single additional atom,
\begin{equation}\label{eq:arbel_measure_directsim}
G_{\varepsilon}
=
\sum_{j=1}^{\tau(\varepsilon)} w_j \, \delta_{\theta_j}
\;+\;
R_{\tau(\varepsilon)} \, \delta_{\theta_0},
\qquad \theta_0 \sim G_0 .
\end{equation}
By construction, $R_{\tau(\varepsilon)} < \varepsilon$ almost surely, and hence $d_{\mathrm{TV}}(G,G_{\varepsilon})\le R_{\tau(\varepsilon)} < \varepsilon$ a.s., with 
\[
d_{\mathrm{TV}}(\mu,\nu)
:= \sup_{A\in\mathcal B_{\mathbb X}} |\mu(A)-\nu(A)|
= \tfrac12 \|\mu-\nu\|_{1}.
\]
In the DP case, where 
$v_j \overset{\text{iid}}{\sim} \mathrm{Beta}(1,\alpha)$,
a further probabilistic characterization is available.
If $v \sim \mathrm{Beta}(1,\alpha)$, then
$Y := -\log(1-v) \sim \mathrm{Exp}(\alpha)$, so that
\begin{equation}\label{eq:exp_wait_directsim}
-\log R_n = \sum_{j=1}^n Y_j .
\end{equation}
Consequently, $\tau(\varepsilon)$ can be interpreted in terms of a Poisson
process in ``time'' $t = \log(1/\varepsilon)$, 
\begin{equation}\label{eq:pois_tau_directsim}
\tau(\varepsilon) - 1 \sim \mathrm{Pois}\big(\alpha \log(1/\varepsilon)\big).
\end{equation}
This representation quantifies expected effort for accuracy $\varepsilon$ under almost sure control.

Theorem~\ref{thm:finite-ssrm} provides a complementary route, in a more general setup.  Instead of choosing a truncation level to make $R_n$ smaller than a pre-specified tolerance, it introduces a latent truncation variable $K$ with a fully specified  distribution given the weights, and reweights the first $K$ atoms.
Given a choice of $\boldsymbol{\xi}$  and $\boldsymbol{w}$, direct simulation proceeds by first sampling  $K$ from the pmf in Theorem~\ref{thm:finite-ssrm} and then \eqref{finiterep}.
Figure~\ref{fig:dp_directsim_cdf}
illustrates such simulations for the DP with $G_0=\mathrm{Unif}(0,1)$ and deterministic decreasing sequence  $\xi_j=\exp(-\eta j)$, with $\eta>0$. The figure also shows the simulations corresponding to the truncation with almost sure error control \citet{ArbelDeBlasiPrunster2019} and a large fixed truncation at $N$ (e.g.\ $N=10^4$) with tail lumping.

\begin{figure}[H]
  \centering
  \includegraphics[width= 0.8\textwidth]{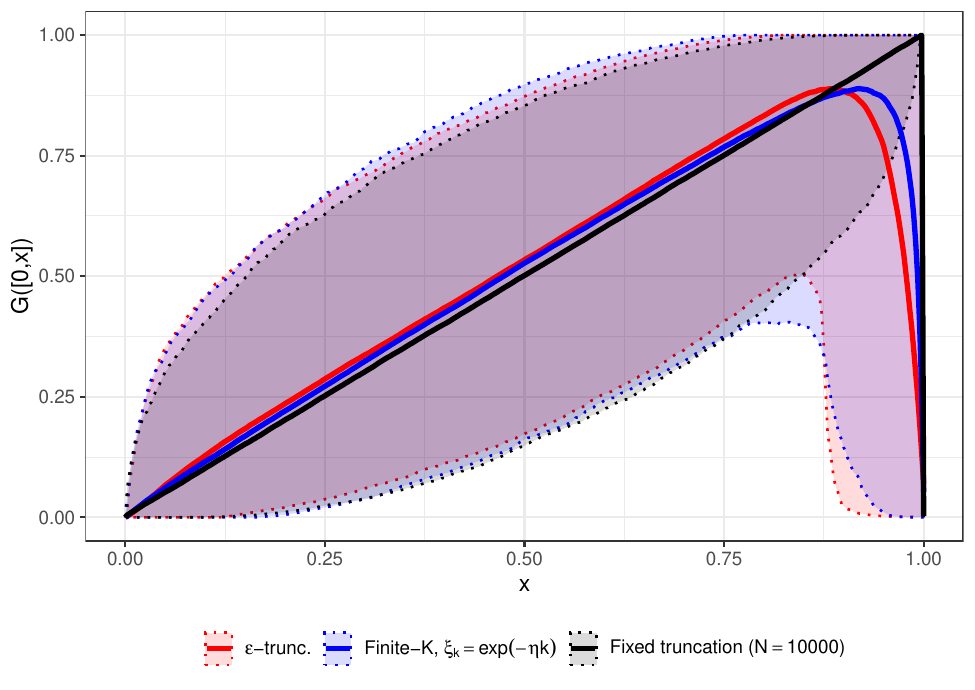}
  \caption{DP simulation comparison under three  scenarios.
  The figure displays the pointwise mean of $G([0,x])$ across repeated simulations and pointwise $95\%$ bands.
  }
  \label{fig:dp_directsim_cdf}
\end{figure}

\subsection{Total variation comparisons}\label{subsec:tv-directsim}

A convenient way to compare truncation schemes, and to assess the fidelity of
finite representations relative to our construction, is via the total variation
distance. This choice is especially natural for coupling arguments, since
competing finite measures can be built from the \emph{same} underlying
stick--breaking realization. Throughout this section, all comparisons are made
under the natural coupling in which the competing finite measures are
constructed from a common series representation $(w_j,\theta_j)_{j\ge1}$.

A key technical ingredient is to control the discrepancy induced by the
$\xi$--reweighting of the first $k$ atoms.
This operation can be viewed as a form of \emph{tilting} of a discrete
probability mass function. The following uniform bound will be used repeatedly.

\begin{lemma}\label{lem:tilt_tv}
Let $P$ be a probability mass function on $\{1,\ldots,k\}$ and let $f:\{1,\ldots,k\}\to(0,\infty)$ satisfy $0<a\le f(i)\le b<\infty$ for all $i$. Define a tilted  $Q$ by
\[
Q(i):=\frac{f(i)P(i)}{\mathbb E_P[f]},\qquad i=1,\ldots,k.
\]
Then
\[
d_{\mathrm{TV}}(P,Q)\le \frac{b-a}{b+a}
=\frac{(b/a)-1}{(b/a)+1}.
\]
\end{lemma}
\begin{proof}
Let $m := \mathbb E_P[f] \in [a,b]$ and define $g(i):=f(i)/m$.
Then $\mathbb E_P[g]=1$ and
\[
g(i)\in\Big[\frac{a}{b},\frac{b}{a}\Big]
=
\Big[\frac{1}{r},r\Big],
\qquad r:=\frac{b}{a}\ge 1.
\]
Since $Q(i)=g(i)P(i)$, we have
\[
d_{\mathrm{TV}}(P,Q)
=
\frac12 \sum_{i=1}^k P(i)\,|1-g(i)|
=
\frac12\,\mathbb E_P\big[|1-g|\big].
\]
Under $g\in[1/r,r]$ with $\mathbb E_P[g]=1$, the maximum of $\mathbb E_P|1-g|$ is attained by a two-point law on $\{1/r,r\}$, giving $\mathbb E_P|1-g|\le 2(r-1)/(r+1)$ and hence
\[
d_{\mathrm{TV}}(P,Q)\le \frac{r-1}{r+1}=\frac{b-a}{b+a}.
\]
\end{proof}
The next proposition presents  useful total variation bounds for our representation.
\begin{proposition}\label{prop:tv_bounds}

Let $G(\cdot)=\sum_{j\ge1} w_j\,\delta_{\theta_j}(\cdot)$ be a proper SSP and set $R_k:=\sum_{j>k}w_j$.
Let $\boldsymbol\xi=(\xi_j)_{j\ge1}$ be strictly decreasing with $\xi_j>0$ and $\xi_j\downarrow0$ a.s.
For $k\ge1$ define
\[
s_k:=\sum_{h=1}^k \frac{w_h}{\xi_h},\qquad
\tilde w_j:=\frac{w_j/\xi_j}{s_k},\qquad
G^\star_k:=\sum_{j=1}^k \tilde w_j\,\delta_{\theta_j},
\]
and define the renormalized truncation
\[
G^{(k)}_{\mathrm{ren}}:=\sum_{j=1}^k \frac{w_j}{1-R_k}\,\delta_{\theta_j}.
\]
Equivalently, if $(\bar w_j)_{j\ge1}$ denotes the full weight sequence of $G^{(k)}_{\mathrm{ren}}$ on the
countable support $\{\theta_j:j\ge1\}$, where $\bar w_j=w_j/(1-R_k)$ for $j\le k$ and $\bar w_j=0$ for $j>k$, with $(w_j)_{j>k}$ the original tail weights of $G$. Let $M_k:=\xi_1/\xi_k$ and $D_k:=(M_k-1)/(M_k+1)$. Then, for every $k\ge1$,

\begin{enumerate}
\item $d_{\mathrm{TV}}(G^{(k)}_{\mathrm{ren}},G^\star_k)\le D_k$ and $d_{\mathrm{TV}}(G,G^\star_k)\le R_k+D_k$.
\item If $\xi_j=e^{-\eta j}$ with $\eta>0$, then $D_k=\tanh\!\big(\eta(k-1)/2\big)$.
\item  Suppose $w_j=v_j\prod_{\ell<j}(1-v_\ell)$ and define the
remaining stick after $k-1$ breaks by
\[
T_{k-1}:=\prod_{\ell<k}(1-v_\ell),\qquad T_0:=1.
\]
Choose $\xi_k:=T_{k-1}$. Since $T_{k-1}=\sum_{j\ge k}w_j=R_{k-1}$, we have
\[
D_k=\frac{M_k-1}{M_k+1}
=\frac{(1/T_{k-1})-1}{(1/T_{k-1})+1}
=\frac{1-T_{k-1}}{1+T_{k-1}}
=\frac{1-R_{k-1}}{1+R_{k-1}}.
\]

\item If $K$ follows the distribution from Theorem~1, 
\[
\mathbb E_{K\mid \boldsymbol w,\boldsymbol\xi}\!\big[d_{\mathrm{TV}}(G,G^\star_K)\big]
\le
\mathbb E_{K\mid \boldsymbol w,\boldsymbol\xi}[R_K]
+
\mathbb E_{K\mid \boldsymbol w,\boldsymbol\xi}[D_K].
\]
If, in addition, $\boldsymbol\xi$ is deterministic (or random but independent of $\boldsymbol w$), then
\[
\mathbb P(K=k)
=
(\xi_k-\xi_{k+1})\sum_{h=1}^k\frac{\mathbb E[w_h]}{\xi_h},
\qquad
\mathbb E[R_K]
=
\sum_{1\le h<j}\Big(1-\frac{\xi_j}{\xi_h}\Big)\,\mathbb E[w_h w_j],
\]
and therefore
\begin{equation}\label{eq:tv_bound_deterministic_xi}
\mathbb E\!\big[d_{\mathrm{TV}}(G,G^\star_K)\big]
\le
\sum_{1\le h<j}\Big(1-\frac{\xi_j}{\xi_h}\Big)\,\mathbb E[w_h w_j]
+
\sum_{k\ge1} D_k\,\mathbb P(K=k).
\end{equation}
Moreover, if $G_\varepsilon$ is a coupled truncation with $d_{\mathrm{TV}}(G,G_\varepsilon)<\varepsilon$ a.s., then
\[
d_{\mathrm{TV}}(G^\star_K,G_\varepsilon)\le \varepsilon+R_K+D_K\quad\text{a.s.},
\qquad
\mathbb E[d_{\mathrm{TV}}(G^\star_K,G_\varepsilon)]
\le \varepsilon+\mathbb E[R_K]+\mathbb E[D_K].
\]
\end{enumerate}
\end{proposition}

\begin{remark}\label{rem:random_xi_as}
Let $\Omega_\xi$ be as in the proof of Theorem~\ref{thm:finite-ssrm}. 
Since the bounds in Proposition~\ref{prop:tv_bounds} are \emph{pathwise} in $(\boldsymbol w,\boldsymbol\xi)$, they hold deterministically on $\Omega_\xi$. Hence, if $\boldsymbol\xi$ is random with $\mathbb P(\Omega_\xi)=1$, the same inequalities hold a.s.  under the joint law of all random quantities.%
\end{remark}

\begin{proof}
Work on $\Omega_\xi$ and fix $\omega\in\Omega_\xi$. Then the sequence $(\xi_j(\omega))_{j\ge1}$ is deterministic, so all steps below are deterministic for this fixed $\omega$ and hence the  inequalities hold a.s.

\smallskip
\noindent\textbf{(1)} Fix $k\ge1$ and condition on $(\theta_j)_{j=1}^k$, so both $G^{(k)}_{\mathrm{ren}}$ and $G^\star_k$ have the same support.
Let $P(j):=\bar w_j=w_j/(1-R_k)$ for $j=1,\dots,k$ and set $f(j):=\xi_j^{-1}$. Then $\tilde w_j=f(j)P(j)/{\mathbb E_P[f]},$
so $Q(j):=\tilde w_j$ is the tilt of $P$ by $f$. Since $\xi_1>\cdots>\xi_k>0$, we have
$f(j)\in[1/\xi_1,\,1/\xi_k]$ and $b/a=(1/\xi_k)/(1/\xi_1)=\xi_1/\xi_k=M_k$.
Lemma~\ref{lem:tilt_tv} yields
\[
d_{\mathrm{TV}}\!\big(G^{(k)}_{\mathrm{ren}},G^\star_k\big)=d_{\mathrm{TV}}(P,Q)\le \frac{M_k-1}{M_k+1}=D_k.
\]
We now show that $d_{\mathrm{TV}}(G,G^{(k)}_{\mathrm{ren}})=R_k$.
Since both measures are supported on the countable set $\{\theta_j:j\ge1\}$, total variation reduces to half the $\ell^1$ distance between the corresponding weight sequences
$d_{\mathrm{TV}}(G,G^{(k)}_{\mathrm{ren}})=\frac12\sum_{j\ge1}|w_j-\bar w_j|$, where $\bar w_j=\frac{w_j}{1-R_k}$ for $j\le k$ and $\bar w_j=0$ for $j>k$. Hence
\[
\sum_{j\ge1}\big|w_j-\bar w_j\big|
=
\sum_{j\le k}\Big|w_j-\frac{w_j}{1-R_k}\Big|+\sum_{j>k}|w_j-0|
=
\frac{R_k}{1-R_k}\sum_{j\le k}w_j+\sum_{j>k}w_j
= R_k+R_k =
2R_k,
\]
and  $d_{\mathrm{TV}}(G,G^{(k)}_{\mathrm{ren}})=R_k$. Thus, by the triangle inequality, we obtain inequality (1)
\[
d_{\mathrm{TV}}(G,G^\star_k)
\le d_{\mathrm{TV}}(G,G^{(k)}_{\mathrm{ren}})+d_{\mathrm{TV}}(G^{(k)}_{\mathrm{ren}},G^\star_k)
=R_k+D_k.
\]

\smallskip
\noindent\textbf{(2)} If $\xi_j=e^{-\eta j}$, then $M_k=\xi_1/\xi_k=e^{\eta(k-1)}$, and using $(e^x-1)/(e^x+1)=\tanh(x/2)$
\[
D_k=\frac{e^{\eta(k-1)}-1}{e^{\eta(k-1)}+1}
=\tanh\!\Big(\frac{\eta(k-1)}{2}\Big).
\]

\smallskip
\noindent\textbf{(3)} 
Let $T_{k-1}:=\prod_{\ell<k}(1-v_\ell)$ with $T_0=1$, and take $\xi_k:=T_{k-1}$.
Since $T_{k-1}=\sum_{j\ge k}w_j=R_{k-1}$, we have $M_k=1/T_{k-1}$ and hence
\[
D_k=\frac{M_k-1}{M_k+1}=\frac{1-T_{k-1}}{1+T_{k-1}}=\frac{1-R_{k-1}}{1+R_{k-1}}.
\]

\smallskip
\noindent\textbf{(4)} On $\{K=k\}$, item~(1) gives $d_{\mathrm{TV}}(G,G^\star_K)\le R_K+D_K$.
Taking conditional expectation with respect to $K\mid(\boldsymbol w,\boldsymbol\xi)$ yields
\begin{equation}\label{eq:tv_avg_over_K}
\mathbb E_{K\mid \boldsymbol w,\boldsymbol\xi}\!\big[d_{\mathrm{TV}}(G,G^\star_K)\big]
\le
\mathbb E_{K\mid \boldsymbol w,\boldsymbol\xi}[R_K]
+
\mathbb E_{K\mid \boldsymbol w,\boldsymbol\xi}[D_K].
\end{equation}
If $G_\varepsilon$ is a coupled truncation with $d_{\mathrm{TV}}(G,G_\varepsilon)<\varepsilon$ a.s., then
\[
d_{\mathrm{TV}}(G^\star_K,G_\varepsilon)
\le d_{\mathrm{TV}}(G^\star_K,G)+d_{\mathrm{TV}}(G,G_\varepsilon)
\le (R_K+D_K)+\varepsilon
\quad\text{a.s.},
\]
and taking expectations gives $\mathbb E[d_{\mathrm{TV}}(G^\star_K,G_\varepsilon)]
\le \varepsilon+\mathbb E[R_K]+\mathbb E[D_K]$.

\smallskip
If $\boldsymbol\xi$ is deterministic (or independent of $\boldsymbol w$), then we may also take expectation with respect to the weights.
First,
\begin{equation}
    \mathbb P(K=k)
=
\mathbb E\big[\mathbb P(K=k\mid\boldsymbol w,\boldsymbol\xi)\big]
=
(\xi_k-\xi_{k+1})\,\mathbb E\!\Big[\sum_{h=1}^k\frac{w_h}{\xi_h}\Big]
=
(\xi_k-\xi_{k+1})\sum_{h=1}^k\frac{\mathbb E[w_h]}{\xi_h}.\label{eq:PK_detxi}
\end{equation}
Second,
\[
\mathbb E[R_K]
=
\mathbb E\!\Big[\sum_{k\ge1}\mathbb P(K=k\mid\boldsymbol w,\boldsymbol\xi)\,R_k\Big]
=
\sum_{k\ge1}(\xi_k-\xi_{k+1})
\sum_{h\le k}\sum_{j>k}\frac{\mathbb E[w_h w_j]}{\xi_h}.
\]
Interchanging sums and using $\sum_{k=h}^{j-1}(\xi_k-\xi_{k+1})=\xi_h-\xi_j$ yields
\[
\mathbb E[R_K]
=
\sum_{1\le h<j}\Big(1-\frac{\xi_j}{\xi_h}\Big)\,\mathbb E[w_h w_j].
\]
Combining this identity with $\mathbb E[D_K]=\sum_{k\ge1}D_k\,\mathbb P(K=k)$ gives \eqref{eq:tv_bound_deterministic_xi}.
\end{proof}

The bounds in Proposition~\ref{prop:tv_bounds} are \emph{coupling-based} and \emph{pathwise}. Conditional on a realization of $(\boldsymbol w,\boldsymbol \xi)$ they give deterministic upper bounds on setwise discrepancies, decomposing the latter into two interpretable components: the \emph{tail mass} $R_k$, which is the usual truncation term, and the \emph{$\xi$-distortion} $D_k=(M_k-1)/(M_k+1)$, where $M_k=\xi_1/\xi_k$ quantifies the range of the multipliers $\xi_j^{-1}$ over $\{1,\dots,k\}$. The distortion term would vanish when $\xi_1=\cdots=\xi_k$ (no reweighting) and increases as $\xi_k$ becomes small relative to $\xi_1$.

Accordingly, $d_{\mathrm{TV}}(G,G^\star_k)\le R_k+D_k$ separates tail error from the additional discrepancy induced by ${\xi}$, while $d_{\mathrm{TV}}(G^{(k)}_{\mathrm{ren}},G^\star_k)\le D_k$ quantifies this reweighting effect alone. See Figure \ref{dp_tv_scatter_bounds}, for an empirical validation of these TV upper bounds.
\begin{figure}[H]
  \centering
  \includegraphics[width= 0.7\textwidth]{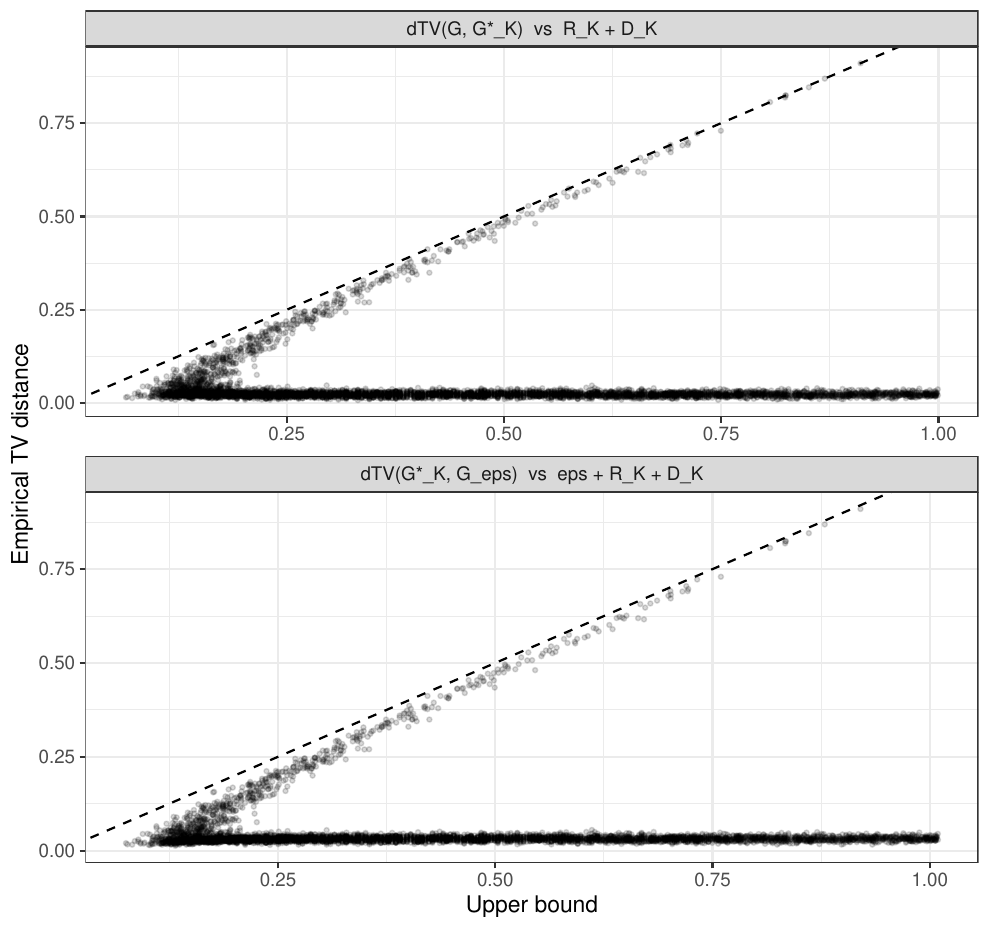}
  \caption{Empirical validation of the TV bounds: $d_{\mathrm{TV}}(G,G^\star_k)$ and $d_{\mathrm{TV}}(G,G_\varepsilon)$. Plots are based on 5,000 simulations with $\alpha = 6$, $\varepsilon=0.01$ and $\eta = 0.01$.
  }
  \label{dp_tv_scatter_bounds}
\end{figure}

When $\boldsymbol\xi$ is deterministic, or random but independent of $\boldsymbol w$, item~(4) in
Proposition~\ref{prop:tv_bounds} can be expressed entirely in terms of moments of the weights.
In particular, combining the identities for $\PP(K=k)$ and $\mathbb{E}[R_K]$ with
$\mathbb{E}[D_K]=\sum_{k\ge1}D_k\,\PP(K=k)$ yields
\begin{equation}\label{eq:tv_bound_detxi_master}
\mathbb E\!\big[d_{\mathrm{TV}}(G,G^\star_K)\big]
\le
\mathbb E[R_K]
+
\sum_{k\ge1} D_k\,\mathbb P(K=k),
\qquad
D_k=\frac{M_k-1}{M_k+1},\ \ M_k=\frac{\xi_1}{\xi_k}.
\end{equation}
For the exponential sequence $\xi_k=q^k=e^{-\eta k}$ ($q=e^{-\eta}\in(0,1)$), one has
$D_k=\frac{1-q^{k-1}}{1+q^{k-1}}=\tanh(\eta(k-1)/2)$ and the expressions for $\PP(K=k)$ often admit a
\emph{matched-rate} removable limit. To streamline notation, define
\[
\Delta_k(x,y)
:=\frac{x^k-y^k}{x-y}\,\mathbb{I}\{x\neq y\}+k\,y^{k-1}\,\mathbb{I}\{x=y\},
\qquad
\Delta_k(y,y)=\lim_{x\to y}\frac{x^k-y^k}{x-y}.
\]
where the case $x=y$ is understood by continuous extension (a removable limit as $x\to y$).

\begin{proposition}\label{prop:K_tail_inherit}
Assume $\xi_k=q^k$ with $q\in(0,1)$ and that $\boldsymbol\xi$ is deterministic (or independent of the weights), so that
$\PP(K=k)=(1-q)\sum_{h=1}^k \mathbb{E}[w_h]\,q^{k-h}$.
If $\mathbb{E}[w_k]\sim C\,k^{-p}$ as $k\to\infty$ for some $C>0$ and $p>0$, then
$\PP(K=k)\sim C\,k^{-p}$ and therefore $\mathbb{E}[K]<\infty$ if and only if $p>2$.
\end{proposition}

\begin{proof}
Write $\PP(K=k)=(1-q)\sum_{m=0}^{k-1}\mathbb{E}[w_{k-m}]\,q^m$.
Since $\mathbb{E}[w_{k-m}]/\mathbb{E}[w_k]\to 1$ for each fixed $m$ and $\sum_{m\ge0}(1-q)q^m=1$,
dominated convergence gives $\PP(K=k)\sim \mathbb{E}[w_k]\sim Ck^{-p}$.
The moment criterion follows from $\sum_k k\,\PP(K=k)<\infty \iff \sum_k k^{1-p}<\infty$.
\end{proof}

{\it Dirichlet process.}
Let $G\sim\mathrm{DP}(\alpha,G_0)$ and set $a:=\alpha/(\alpha+1)$.
For $\xi_k=q^k=e^{-\eta k}$,
\[
\PP(K=k)=\frac{1-q}{\alpha+1}\,\Delta_k(q,a),\qquad k\ge1,
\]
i.e.\ $\PP(K=k)=\frac{1-q}{(\alpha+1)(q-a)}(q^k-a^k)$ for $q\neq a$, with the matched-rate limit
$\PP(K=k)=\frac{k\,a^{k-1}}{(\alpha+1)^2}$ when $q=a$.
Moreover,
\[
\mathbb{E}[R_K]=\frac{\alpha}{2}\,\frac{1-q}{(\alpha+1)-\alpha q},
\]
so \eqref{eq:tv_bound_detxi_master} becomes explicit with $D_k=\tanh(\eta(k-1)/2)$ and the above $\PP(K=k)$.
A convenient calibration is the matched rate $q\approx a$, i.e.
$e^{-\eta}\approx \alpha/(\alpha+1)$ (equivalently $\eta\approx \log(1+1/\alpha)$), which balances the typical size of $K$
against the distortion term $D_k$.

{\it Geometric weights.}
Let $w_k=V(1-V)^{k-1}$, $k\ge1$, and set $a_V:=1-V$.
For $\xi_k=q^k$ and fixed $V$,
\[
\PP(K=k\mid V)=V\,\frac{1-q}{q}\,\Delta_k(q,a_V),\qquad k\ge1,
\]
i.e.\ $\PP(K=k\mid V)=V\frac{1-q}{q-a_V}(q^k-a_V^k)$ for $q\neq a_V$, with the matched-rate limit
$\PP(K=k\mid V)=k\,V^2a_V^{k-1}$ when $q=a_V$.
The stick--breaking choice $\xi_k := T_{k-1} = a_V^{k-1}$
corresponds precisely to the matched rate $q=a_V$ and yields
$\tilde w_j\equiv 1/k$ for $j\le k$ and $K-1\sim\mathrm{NegBin}(2,V)$, so that
$\mathbb{E}[K\mid V]=2/V-1$ and $\mathbb{E}[R_K\mid V]=(1-V)/(2-V)^2$.
If $V\sim\mathrm{Beta}(a_0,b_0)$, the natural-choice pmf marginalizes as
\[
\PP(K=k)=k\,\frac{B(a_0+2,b_0+k-1)}{B(a_0,b_0)}.
\]

{\it Two-parameter Pitman--Yor.}
Let $G\sim\mathrm{PY}(\sigma,\alpha,G_0)$ with $\sigma\in(0,1)$ and $\alpha>-\sigma$.
Under deterministic $\xi_k=q^k$, the general identities apply once $\mathbb{E}[w_k]$ and $\mathbb{E}[w_hw_j]$ are specified.
In particular, using the rising factorial $(x)_n=\Gamma(x+n)/\Gamma(x)$,
\[
\mathbb E[w_k]
=
\frac{1-\sigma}{\alpha+1+(k-1)\sigma}
\prod_{\ell=1}^{k-1}\frac{\alpha+\ell\sigma}{\alpha+1+(\ell-1)\sigma}=
\frac{1-\sigma}{\sigma}\,
\frac{\big(\frac{\alpha}{\sigma}+1\big)_{k-1}}{\big(\frac{\alpha+1}{\sigma}\big)_{k}}.
\]
Moreover, $\mathbb{E}[w_hw_j]$ admits a closed product  computable form in terms of Beta moments.
The tail regime is explicit: $\mathbb{E}[w_k]\sim C_{\sigma,\alpha}\,k^{-1/\sigma}$ as $k\to\infty$, and therefore by
Proposition~\ref{prop:K_tail_inherit},
\[
\PP(K=k)\sim C_{\sigma,\alpha}\,k^{-1/\sigma},
\qquad
\mathbb{E}[K]<\infty\ \Longleftrightarrow\ \sigma<\tfrac12.
\]
For the stick--breaking choice $\xi_k := T_{k-1}$, the conditional simplifications in
Corollary~1 apply (including $D_k=(1-T_{k-1})/(1+T_{k-1})$), but the moment reductions
leading to \eqref{eq:tv_bound_detxi_master} are no longer available because $\boldsymbol\xi$ is adapted to the weights.

In general, when $\xi_k := T_{k-1}$ one has $\xi_k-\xi_{k+1}=w_k$ and $s_k=\sum_{j\le k} w_j/\xi_j$
simplifies (e.g.\ to $\sum_{j\le k} v_j$ under Corollary~1), making simulation and posterior
computation particularly simple.
However, since $\boldsymbol\xi$ depends on the same random variables that determine $(w_j)$, closed-form marginal expressions for $\PP(K=k)$ and $\mathbb{E}[R_K]$ are typically unavailable beyond special cases (such as geometric weights), and are best assessed conditionally (e.g.\ within an augmented sampler) via the diagnostics in item~(4) of Proposition~\ref{prop:tv_bounds}.

Equation \eqref{eq:tv_bound_detxi_master} separates an expected tail term $\mathbb{E}[R_K]$ from an expected distortion term $\sum_k D_k\PP(K=k)$.
For exponential $\xi_k=e^{-\eta k}$, $D_k=\tanh(\eta(k-1)/2)$ gives a transparent tuning knob, while $\PP(K=k)$  is determined by the prior on the weights.
In light-tailed priors (DP and geometric), $\PP(K=k)$ is explicit and $K$ is light-tailed; in heavier-tailed priors (Pitman--Yor with $\sigma>0$), $K$ inherits a polynomial tail under exponential $\boldsymbol\xi$, yielding distinct computational regimes driven by~$\sigma$.
For the choice $\xi_k=R_{k-1}$, conditional simplifications are strong, but marginal closed forms are generally limited to special cases.

\section{SSP mixture modeling and posterior computation}\label{sec:model}
We now turn to the use of our representation  as a finite-dimensional augmentation for posterior inference in SSP mixtures.
Let $G=\sum_{j\ge1} w_j\,\delta_{\theta_j}$ be a proper SSP prior on the parameter space and consider the mixture model for observations  $\boldsymbol x = (x_i)_{i=1}^n$
\begin{equation}\label{eq:ssp_mixture}
x_i\mid G \overset{\text{iid}}{\sim} f_G(x),\qquad
f_G(x)=\int f(x\mid\theta)\,G(d\theta)
=\sum_{j\ge1} w_j\,f(x\mid\theta_j),
\end{equation}
with $\theta_j\overset{\text{iid}}{\sim}G_0$ and a generic kernel $f(\cdot\mid\theta)$. Posterior computation is based on the latent augmentation implied by
Theorem~\ref{thm:finite-ssrm}.
For each observation $i=1,\ldots,n$, we introduce a latent truncation level
$k_i\in\mathbb{N}$ and an allocation variable $z_i\in\mathbb{N}$ satisfying
$z_i \le k_i$, where $k_i$ determines the number of active mixture components
and $z_i$ indexes the component generating $x_i$. A convenient way to see this is through the hierarchical model
\begin{align*}
\boldsymbol w &\sim p(\boldsymbol w)\quad \mbox{and}\quad 
 \theta_j \sim G_0,\ j=1, 2,\ldots \\[4pt]
k_i \mid \boldsymbol w, \boldsymbol \xi &\sim p(k_i = k \mid \boldsymbol w)
= (\xi_k - \xi_{k+1})\,s_k,\ k=1, 2,\ldots\\[4pt]
z_i \mid k_i,\boldsymbol w &\sim \sum_{j=1}^{k_i} \tilde w_{j}\,\delta_j,\ \tilde w_j
= \frac{\xi_j^{-1} w_j}{s_{k_i}},\ 
j = 1,\dots,k_i, \\[4pt]
x_i \mid z_i,\boldsymbol\theta
&\sim f(\cdot \mid \theta_{z_i}),\ i =1, \ldots, n.
\end{align*}
where $s_k:=\sum_{h=1}^k w_h/\xi_h$. Here, $p(\boldsymbol w)$ denotes the  prior on the SSP weights. 

Let $\boldsymbol x=(x_i)_{i=1}^n$, $\boldsymbol z=(z_i)_{i=1}^n$, $\boldsymbol k=(k_i)_{i=1}^n$,
$\boldsymbol\theta=(\theta_j)_{j\ge1}$ and $\boldsymbol w=(w_j)_{j\ge1}$.
The hierarchical model implies the joint density
\begin{align}
p(\boldsymbol x,\boldsymbol z,\boldsymbol k,\boldsymbol\theta,\boldsymbol w)
&=
p(\boldsymbol w)\,
\prod_{j\ge1} p(\theta_j)\,
\prod_{i=1}^n p(k_i\mid \boldsymbol w,\boldsymbol\xi)\,
p(z_i\mid k_i,\boldsymbol w,\boldsymbol\xi)\,
p(x_i\mid z_i,\boldsymbol\theta)\nonumber\\
&=
p(\boldsymbol w)\,
\Bigg[\prod_{i=1}^n(\xi_{k_i}-\xi_{k_i+1})\,\mathbb I(z_i\le k_i)\Bigg]
\prod_{j=1}^{k^\ast}
\Bigg\{
p(\theta_j)\,
\Big(\frac{w_j}{\xi_j}\Big)^{n_j}
\prod_{i:z_i=j} f(x_i\mid \theta_j)
\Bigg\},
\label{eq:bj_compact}
\end{align}
where $k^\ast=\max_{1\le i\le n}k_i$ and $n_j=\sum_{i=1}^n\mathbb I(z_i=j)$ for $j=1,\ldots,k^\ast$.
The cancellation of the normalizing constants $s_{k_i}$ is explicit in \eqref{eq:bj_compact} and yields 
\[
p(z_i,k_i\mid x_i,\boldsymbol w,\boldsymbol\theta,\boldsymbol\xi)\ \propto\
(\xi_{k_i}-\xi_{k_i+1})\,\mathbb{I}\{z_i\le k_i\}\,
\frac{w_{z_i}}{\xi_{z_i}}\,
f(x_i\mid \theta_{z_i}),
\]
since the normalizing factor $s_{k_i}$ cancels. Two essential posterior updates are immediate
\begin{align}
p(z_i=j\mid k_i,\cdots)
&\propto
\frac{w_j}{\xi_j}\,f(x_i\mid\theta_j),
\qquad j=1,\ldots,k_i, \label{eq:zi_full}\\[3pt]
p(k_i=k\mid z_i,\cdots)
&=
\frac{(\xi_k-\xi_{k+1})\,\mathbb{I}\{k\ge z_i\}}{\xi_{z_i}},
\qquad k=z_i,z_i+1,\ldots \label{eq:ki_full}
\end{align}
(where $\sum_{k=z_i}^\infty(\xi_k-\xi_{k+1})=\xi_{z_i}$ is used in \eqref{eq:ki_full}).
Conditionally on the allocations, the component parameters have the standard update
\[
p(\theta_j\mid \boldsymbol x,\boldsymbol z,\boldsymbol k)\ \propto\ p(\theta_j)\prod_{i:z_i=j} f(x_i\mid\theta_j),
\qquad j=1,\ldots,k^\ast,
\]
which is conjugate in the examples of Section~\ref{sec:illustrations}.

Clearly, full conditionals are determined by the choice of the decreasing sequence $\{\xi_j\}$:

{\bf Case A (endogenous $\boldsymbol\xi$).}
When $\xi_j := T_{j-1}=\prod_{\ell<j}(1-v_\ell)$, the updates simplify
substantially. In particular, under Corollary~\ref{C2} one has $\xi_j^{-1}w_j=v_j$ and
$\xi_k-\xi_{k+1}=w_k$, so \eqref{eq:zi_full} becomes $p(z_i=j\mid k_i,\cdots)\propto v_j f(x_i\mid\theta_j)$
and \eqref{eq:ki_full} reduces to a tail-scan over $(w_k)_{k\ge z_i}$. Moreover, the stick-breaking variables admit closed-form Beta updates that incorporate both allocation counts and truncation counts.

{\bf Case B (exogenous $\boldsymbol\xi$).}
When $\xi_j$ is deterministic and independent of the SSP, the truncation update \eqref{eq:ki_full} is particularly convenient: it depends only on the interval lengths $(\xi_k-\xi_{k+1})$ and admits efficient inversion.
For example, for $\xi_j=\exp(-\eta j)$ one obtains the closed-form update
$k_i=\big\lfloor z_i-\eta^{-1}\log U\big\rfloor$ with $U\sim\mathrm{Unif}(0,1)$, while geometric-type sequences yield $k_i=z_i+S_i$ with $S_i$ geometric.
In this regime, the stick-breaking updates retain their standard form, since $\{\xi_j\}$ is independent of the weight-generating variables.

The augmentation above can be viewed as a discrete analogue of the slice construction of \citet{Kalli2011}.
For deterministic $\xi$, the truncation variable $k_i$ plays the role of the last ``visible'' component for observation $i$ and can be sampled directly from \eqref{eq:ki_full} (often in closed form), avoiding the auxiliary continuous slice variables and the associated bookkeeping. 

\section{Truncation, clustering, and finite mixture models}\label{sec:counts}

Several quantities in SSP mixtures evoke the finite-mixture notion of a ``number of components''.  
In a standard finite mixture, the observations $\boldsymbol x=(x_i)_{i=1}^n$ are modeled via a mixing measure
\[
G=\sum_{j=1}^m w_j\,\delta_{\theta_j},
\]
where $m$ (fixed or random) is shared by the entire sample. 
SSP models feature analogous counts, but each indexes a different object in the hierarchy. 
Keeping these roles distinct is key for interpreting both the finite representation of $G$ and the cluster summaries reported in Section~\ref{sec:illustrations}.

The quantities $K$, $c_n$, and $m$, roughly random truncation, clusters and number of components, may each be described informally as component counts, yet they serve fundamentally different purposes (with $K$ and $m$ closest to the usual finite-mixture usage). 
In particular, they correspond to different notions of ``components'' and therefore admit different interpretations:

\begin{itemize}
\item $K$ is the random finite-representation truncation level in Theorem~\ref{thm:finite-ssrm}. It is a \emph{prior-level auxiliary} variable introduced to obtain a two-stage finite representation of the same nonparametric prior on $G$; the original prior  is recovered setwise after averaging over $K$. It exists prior to modeling or observing data. In other words, it is representational/computational and should not be read as a ``true'' number of clusters. 
\item $c_n:=\sum_{j\ge1}\mathbb I\{n_j>0\}$ is the number of occupied clusters in a sample of size $n$.
It is a \emph{data-dependent} occupancy statistic of the induced random partition, and it can be computed
from the allocations by counting the nonempty cluster sizes $n_j$.

\item $m$ is the number of components in a classical finite mixture. It is a \emph{structural} model dimension: changing $m$ defines a different statistical model, and when identifiable under a finite-mixture specification it is a genuine parameter of interest.
\end{itemize}

For posterior computation in SSP mixtures we introduce per-observation truncation variables $k_i$ and allocations $z_i$ as \emph{algorithmic augmentation}. They restrict the set of components that are ``active'' for each $x_i$ and yield finite-dimensional updates. These variables are auxiliary: their joint distribution is chosen so that, after integrating out $(\boldsymbol z,\boldsymbol k)$, one recovers exactly the same joint law for $(G,x_1,x_2,\ldots)$ as under the original SSP mixture. Accordingly, $k_i$ should not be interpreted as a model parameter or as a proxy for either $K$, $c_n$ or $m$.

In approaches such as those referred to as mixtures of finite mixtures (MFM) \citep{Miller2013inconsistency}, one places a prior on a finite number of components $m$ and, conditional on $m$, the mixing measure is typically a symmetric Dirichlet distribution over $m$ atoms (or a closely related finite-species construction). In this setting, $m$ is a \emph{model-level parameter}: varying its prior changes the prior law of the mixing measure, and one can study relationships such as $\mathbb P(c_n=c\mid m)$ and posterior concentration on a finite ``true'' $m$ under suitable conditions (e.g.\ for finite-species Gibbs-type priors with $\sigma<0$). {See, e.g., \citet{GnedinPitman2006,LijoiMenaPrunster2007Biometrika,DeBlasi2013,miller2018mixture} and references therein.}

In contrast, in Theorem~\ref{thm:finite-ssrm} the random truncation level $K$ is derived from the SSP weights and the chosen decreasing sequence $\{\xi_j\}$; it is introduced only to obtain an exact conditionally finite representation. Averaging over the induced law of $K$ recovers exactly the original SSP prior on $G$ at the level of set masses, so $K$ does not encode an additional modeling choice and does not alter the support of the prior or its induced partition structure.

\subsection{Interpreting \texorpdfstring{$c_n$}{c\_n} in finite and nonparametric models}
\label{subsec:counts_cn}

In nonparametric SSP mixtures, $c_n$ is naturally interpreted as a summary of the random partition induced by $G$ on a sample of size $n$. In finite-mixture models, however, $c_n$ is sometimes informally compared with $m$; apparent discrepancies are then occasionally misinterpreted within interpretations of  Bayesian nonparametric methods. From the viewpoint above, this tension typically reflects \emph{identifiability of $m$ under a chosen finite model} and the fact that $c_n$ is an occupancy statistic rather than a dimension parameter. In Section~\ref{sec:illustrations} we therefore report $c_n$ as a diagnostic summary of posterior partition structure, and we compare it with the data-generating finite-mixture order only for the simulated example and only as a qualitative check, not as a consistency claim about a ``true'' number of components.

The key message  is that: $K$ (representation size), $c_n$ (occupancy), and $m$ (finite-model dimension) are not directly comparable. In particular, it is generally inappropriate to interpret $c_n$ as an estimator of $K$ or to interpret $K$ as a ``true'' number of clusters/components.

\section{Illustrations}\label{sec:illustrations}

We illustrate the impact of our augmentation on mixture-model inference using simulated and real data. We report (i) posterior predictive density estimates and credible bands, (ii) the posterior behaviour of the number of occupied clusters $c_n=\sum_{j\ge1}\mathbb{I}\{n_j>0\}$, with $n_j = \sum_{i=1}^n \mathbb I(z_i = j),\quad \mbox{for}\quad j = 1,2,\ldots, k^\ast$, $k^\ast = \displaystyle\max_{i=1, \ldots, n} \left\{k_i\right\}$ and (iii) execution times.
We compare finite-representation samplers (DPFinite and GSBFinite, corresponding to Corollary~\ref{C2} and the geometric specialization) with generalized slice samplers \citep{Kalli2011} under matched priors (DPSlice and GSBSlice).
For the finite-representation samplers we consider both the endogenous ``natural'' choice $\xi_j=\prod_{\ell<j}(1-v_\ell)$ and the exogenous exponential choice $\xi_j=\exp(-\eta j)$, highlighting how $\eta$ affects both computational cost (through typical truncation levels) and mixing behaviour (through the induced reweighting). For the DPSlice and GSBSlice models we have considered the same deterministic sequence in all the examples by taking $\eta=1$.  For all models, we take Normal kernels so that
$\theta_j = (\mu_j, \tau_j)$, for which we assign independent normal–gamma priors
\begin{equation}
    G_0(\mu_j, \tau_j) = \dnormal(\mu_j \mid \mu_0, \tau_0^{-1}) \, \dgamma(\tau_j \mid a, b).
\end{equation}

Throughout the experiments, the hyperparameters are fixed at
\[
(\mu_0, \tau_0, a, b) = (0, 0.001, 0.001, 0.001),
\]
a weakly informative specification that has minimal influence on the posterior
distribution. For the Dirichlet process–based models, the concentration parameter $\alpha$ is assigned a Gamma prior, $\alpha \sim \dgamma(0.1, 0.1)$, which has mean $1$ and variance $10$. For the GSBFinite and GSBSlice models, we assign a uniform prior $v \sim {\rm Beta}(1, 1)$ for the geometric parameter. All samplers
were run for $S=100,\!000$ iterations, with predictive samples obtained after a
burn–in period of $20,\!000$ iterations.

For each example, we additionally report the execution times of all the models in the comparison over the $100,\!000$ iterations for the different choices of $\xi_j.$ All simulations have been conducted using the Julia language on a MacBook Air with M2 chip and 8GB RAM.
Implementation details and the Julia code required to reproduce all numerical results are provided in the Supplementary material.

{\it Monte Carlo predictive density estimator.}
Given posterior draws $\{(\boldsymbol w^{(s)},\boldsymbol\theta^{(s)})\}_{s=1}^S$,
we estimate the predictive density by
\[
\hat f(x)=\frac{1}{S}\sum_{s=1}^S f_{G^{(s)}}(x)
=\frac{1}{S}\sum_{s=1}^S \sum_{j=1}^{k^{\ast(s)}} w_j^{(s)}\,f\!\left(x\mid \theta_j^{(s)}\right),
\]
where $k^{\ast(s)}=\max\{k_1^{(s)},\ldots,k_n^{(s)}\}$ for the finite-representation samplers (and analogously for the slice samplers).
For the Normal–Gamma specification in the experiments, $\theta_j=(\mu_j,\tau_j)$ and
$f(x\mid\theta_j)=\mathcal N\!\big(x;\mu_j,\tau_j^{-1}\big)$.

\subsection{Simulated data example}\label{subsec:simulated}

We first consider  $n=250$ observations from a four–component Gaussian mixture 
\begin{equation}\label{DataMix}
f(x)=\sum_{j=1}^{4} w_j\,\dnormal(x\mid \mu_j,\sigma_j^2),
\end{equation}
with $w_{1:4}=(0.5,0.2,0.2,0.1)$, $\mu_{1:4}=(-4,0,5,8)$ and $\sigma_{1:4}=(0.8,1,0.5,1.5)$.
We compare DP/GSB finite–representation samplers against their slice counterparts, and report (i) posterior predictive density estimates with $95\%$ credible bands, (ii) the evolution of the occupied–cluster count $c_n$, and (iii) execution times.

\begin{figure*}[tbh]
        \centering
        \includegraphics[width=\textwidth]{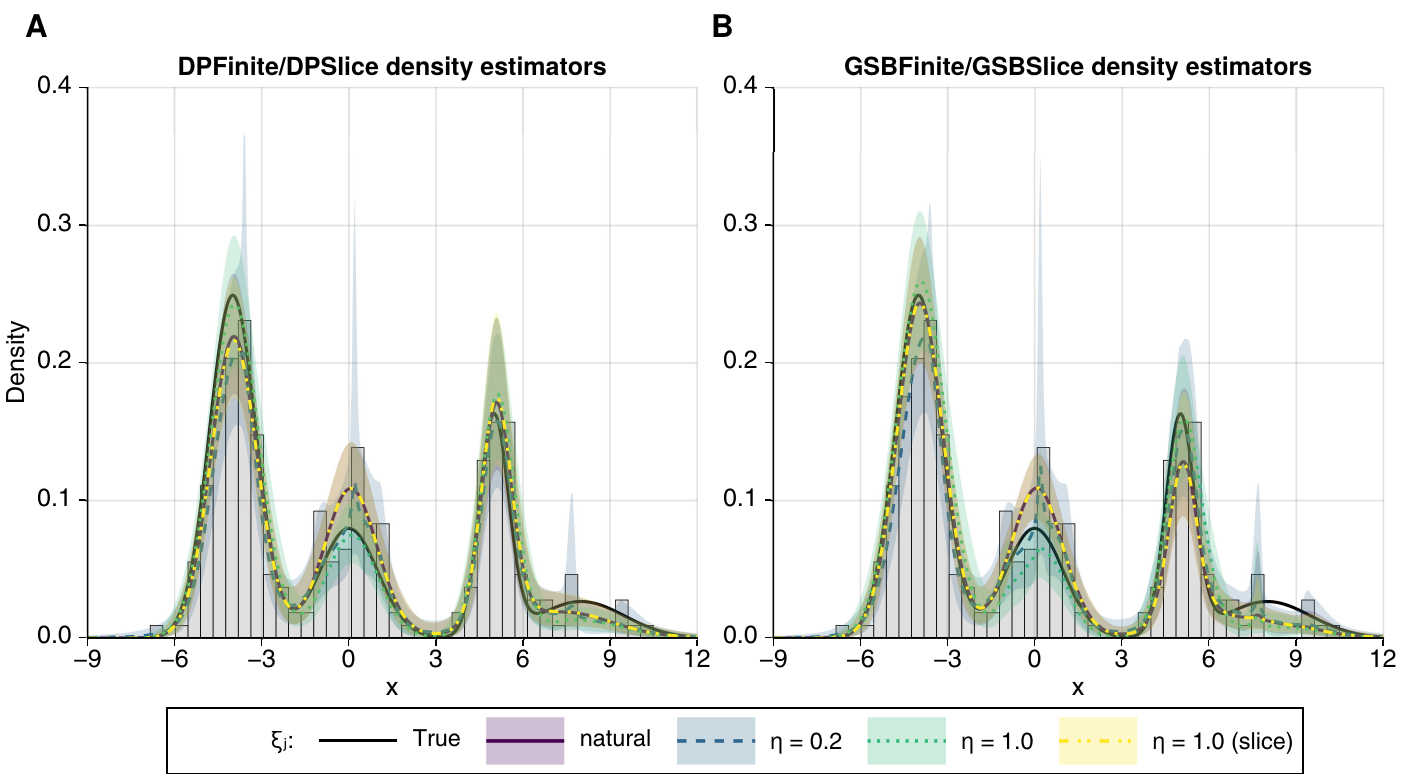}
        \caption{Histogram of the simulated data with Monte Carlo density estimators and $95\%$ credible intervals for different choices of $\xi_j$ and $\eta$. Panel A: DPFinite vs.\ DPSlice. Panel B: GSBFinite vs.\ GSBSlice.}
        \label{fig:fig1}
\end{figure*}

\begin{figure*}[tbh]
        \centering
        \includegraphics[width=\textwidth]{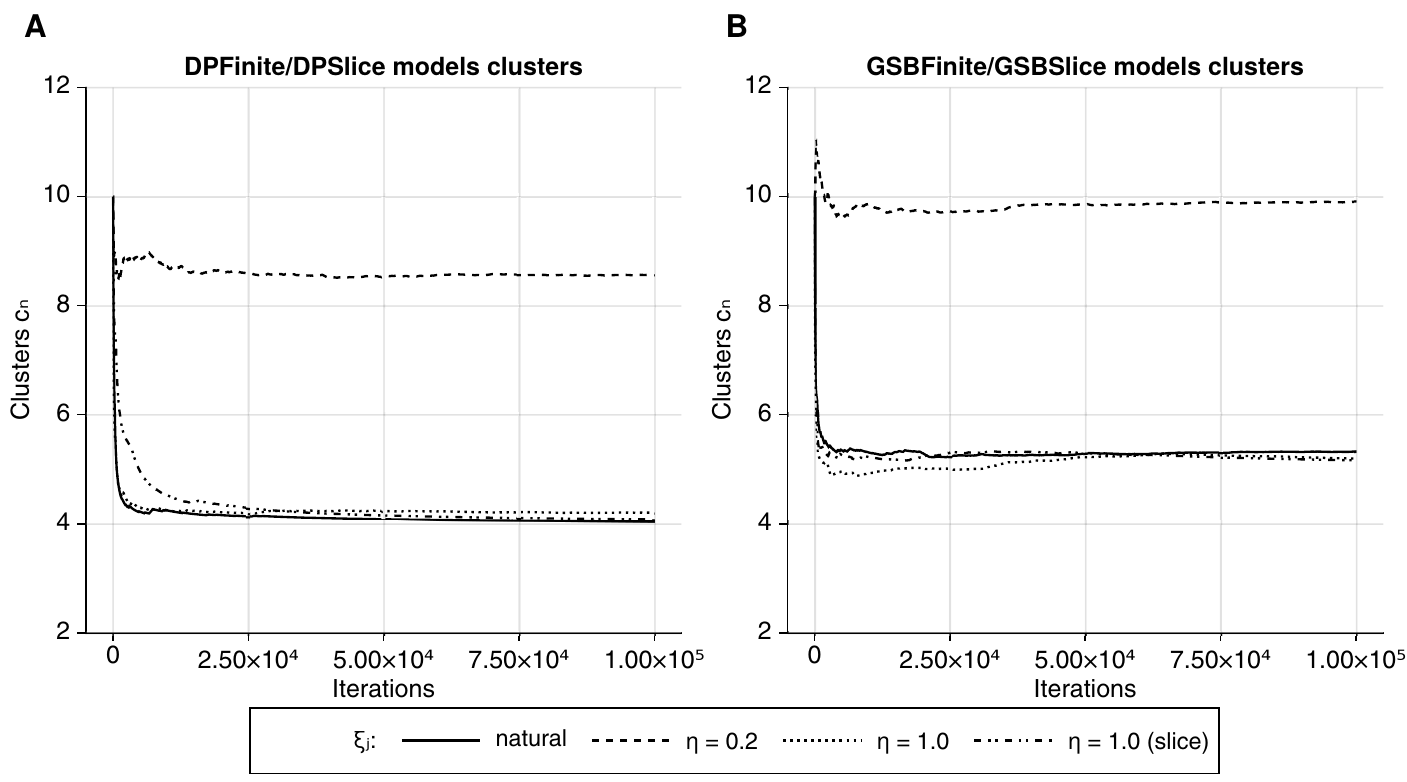}
        \caption{Ergodic means of the occupied–cluster count $c_n$ over iterations. Panel A: DPFinite vs.\ DPSlice. Panel B: GSBFinite vs.\ GSBSlice, for different choices of $\xi_j$ and $\eta$.}
        \label{fig:fig2}
\end{figure*}

Figures~\ref{fig:fig1}--\ref{fig:fig2} summarize the main qualitative behavior.
A small value of $\eta$ (here $\eta=0.2$) tends to increase posterior mass on larger $c_n$ for both DP/GSBFinite, which is reflected in sharper local features of the corresponding density estimates.
For the remaining values of $\eta$ and for the natural random sequence $\xi_j$, DPFinite and DPSlice concentrate around the order of the data-generating mixture.
The GSB-based models yield slightly larger $c_n$, consistent with the over-clustering behavior reported in \citet{DeBlasi2020, Hatjispyros2023}.

\begin{table*}[!htb]
\caption{Execution times (seconds) for $100,\!000$ iterations for the simulated four–component mixture, for two sample sizes.\label{tab:table1}}
\centering
\setlength{\tabcolsep}{6pt}
\begin{tabular*}{\textwidth}{@{\extracolsep{\fill}}lcccc@{}}
\toprule
 & \multicolumn{2}{c}{DP based models} & \multicolumn{2}{c}{GSB based models} \\
\cmidrule(lr){2-3}\cmidrule(lr){4-5}
$\xi_j$ & $n=250$ & $n=1000$ & $n=250$& $n=1000$  \\
\midrule
natural  & 10.508  & 37.697 & 18.429& 74.997 \\
$\eta=0.2$ & 33.976   & 82.733&  35.564& 79.480\\
$\eta =1.0$   & 11.058   & 26.635& 11.192 & 26.868\\
$\eta=1.0$ (slice)  & 20.242  & 64.697 & 21.492 &71.570 \\
\bottomrule
\end{tabular*}
\begin{tablenotes}
\footnotesize
\item Note: density estimates are evaluated on a grid of $500$ points.
\end{tablenotes}
\end{table*}

Table~\ref{tab:table1} reports execution times for $n=250$ and for a larger dataset ($n=1000$) generated from \eqref{DataMix}. In this experiment, DPFinite is consistently faster than DPSlice, with the natural choice of $\xi_j$ providing the most favorable scaling.

\subsection{Galaxy data}\label{subsec:galaxy}

We next analyze the galaxy data: velocities (km/s) of $n=82$ galaxies in the Corona Borealis region, a standard benchmark known to exhibit multimodality with roughly three to six clusters in many analyses \citep{RichardsonGreen1997, RoederWasserman1997}. As before, we focus on posterior predictive density estimates, the occupied–cluster count $c_n$, and execution times.

\begin{figure*}[tbh]
        \centering
        \includegraphics[width=\textwidth]{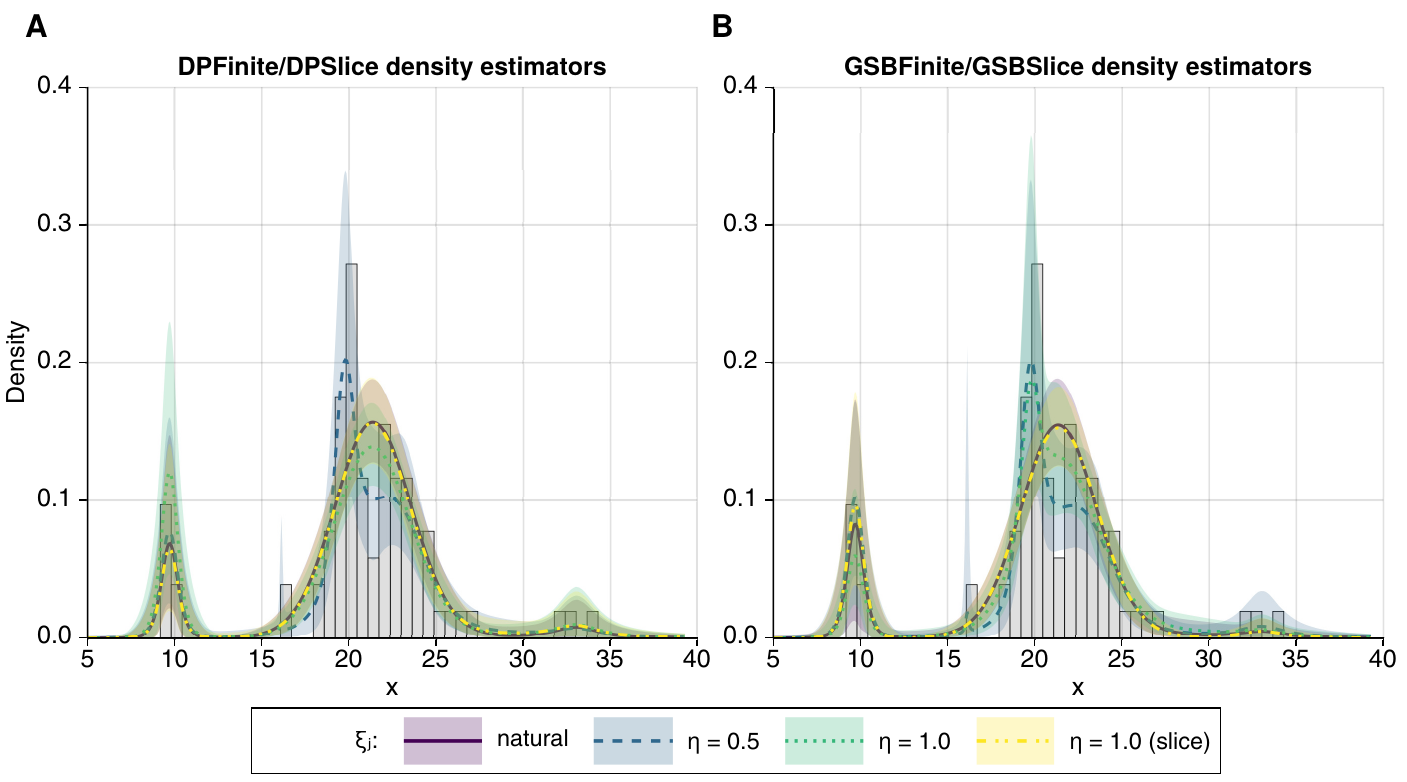}
        \caption{Galaxy data: histogram with Monte Carlo density estimators and $95\%$ credible intervals for different choices of $\xi_j$ and $\eta$. Panel A: DPFinite vs.\ DPSlice. Panel B: GSBFinite vs.\ GSBSlice.}
        \label{fig:fig3}
\end{figure*}

\begin{figure*}[tbh]
        \centering
        \includegraphics[width=\textwidth]{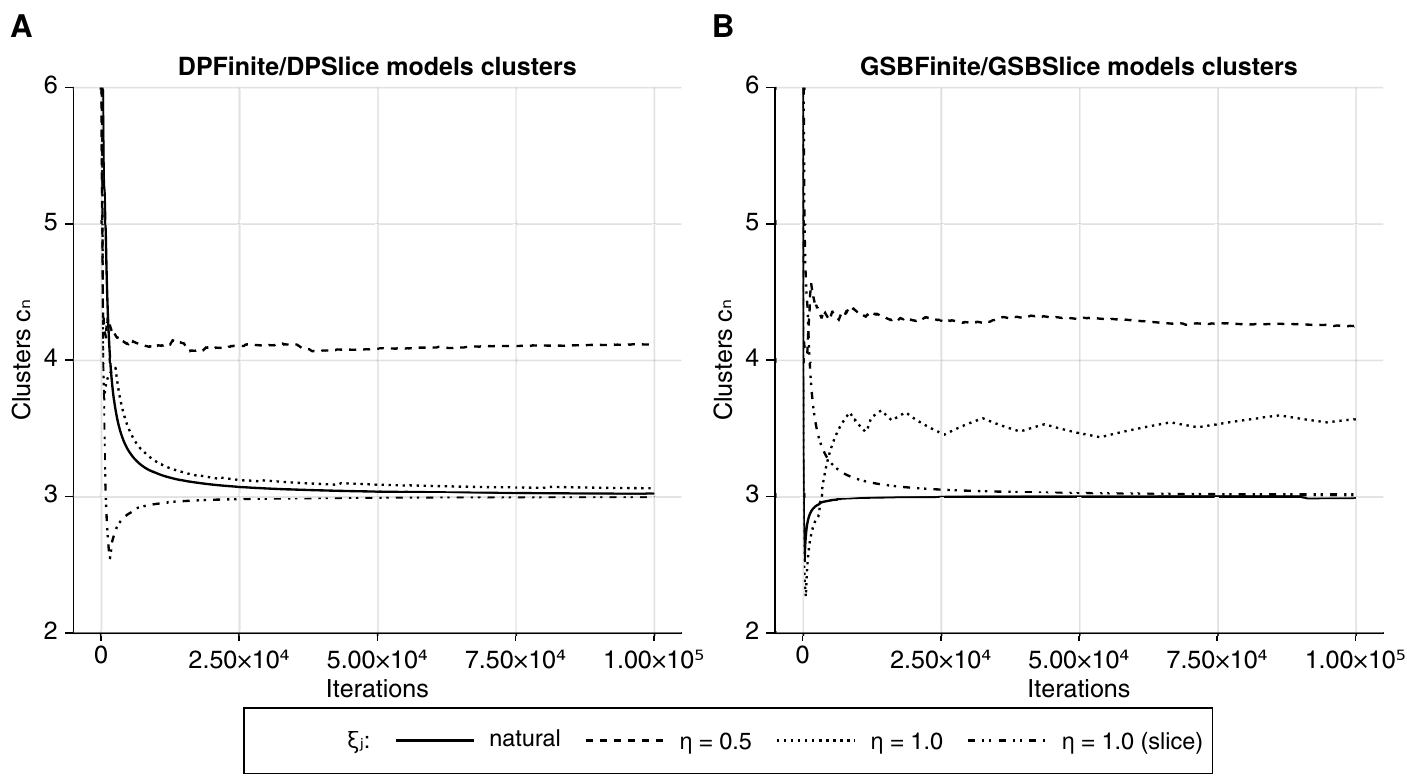}
        \caption{Galaxy data: ergodic means of the occupied–cluster count $c_n$ over iterations. Panel A: DPFinite vs.\ DPSlice. Panel B: GSBFinite vs.\ GSBSlice, for different choices of $\xi_j$ and $\eta$.}
        \label{fig:fig4}
\end{figure*}

Figure~\ref{fig:fig3} shows that all methods capture the multimodal structure of the data.
Figure~\ref{fig:fig4} indicates that the DP-based models stabilize around three occupied clusters, while smaller values of $\eta$ (here $\eta=0.5$) lead to slightly larger posterior $c_n$, reflecting increased exploration of additional components.

\begin{table*}[!htb]
\caption{Execution times (seconds) for $100,\!000$ iterations for the galaxy data.\label{tab:table2}}
\centering
\setlength{\tabcolsep}{6pt}
\begin{tabular*}{\textwidth}{@{\extracolsep{\fill}}lcc@{}}
\toprule
$\xi_j$ & DPFinite/DPSlice & GSBFinite/GSBSlice  \\
\midrule
natural  & 5.269   & 5.827 \\
$\eta=0.5$ &14.449  & 9.434\\
$\eta =1.0$   & 7.056   & 6.428\\
$\eta=1.0$ (slice)  & 7.061 & 6.068\\
\bottomrule
\end{tabular*}
\begin{tablenotes}
\footnotesize
\item Note: density estimates are evaluated on a grid of $500$ points.
\end{tablenotes}
\end{table*}

Execution times are reported in Table~\ref{tab:table2}. In this dataset, the natural choice of $\xi_j$ again yields competitive runtimes, and the finite–representation samplers are comparable to (and in some cases faster than) their slice-based counterparts.

\section{Conclusions}
This paper develops a new perspective on Bayesian nonparametric  modeling by introducing an exact two-stage finite representation of proper species sampling processes. The key insight is that an SSP can be reparametrized via a  finite-mixture construction with a latent truncation level 
$K$ and reweighted atoms, while preserving the original SSP setwise after averaging over 
$K$. This representation is therefore not an approximation but a structural reformulation that yields a  finite random measure together with an explicit  law  for $K$.

Building on this representation, we proposed finite-dimensional augmentation schemes for SSP mixtures and derived Gibbs samplers that avoid ad hoc fixed truncation levels. The resulting updates are simple to implement and apply broadly across SSP priors used in mixture models, including Dirichlet, two-parameter Pitman--Yor, geometric, and more general stick--breaking families, as well as dependent-length constructions.

Beyond MCMC, the two-stage finite representation can also be useful in settings where one needs finite-dimensional prior draws from an SSP. For example, in fast-search methods for BNP mixtures such as BNP--CAEM \citep{Karabatsos2021}, prior draws are typically obtained through truncation; our construction provides an exact two-stage alternative in which the truncation level is random and model-induced, and the original SSP  is recovered setwise after averaging over that auxiliary truncation variable.

Our empirical results on simulated mixtures and the benchmark galaxy dataset show that the proposed finite-representation samplers recover the underlying density and yield sensible posterior distributions for the occupied-cluster count $c_n$, with competitive (and often improved) execution times relative to generalized slice samplers. In particular, the choice of the decreasing sequence $\{\xi_j\}$ can have a noticeable impact on both mixing behaviour and computational cost, with the natural stick-breaking choice often leading to particularly efficient updates.

Several quantities in SSP mixtures resemble a ``number of components'' but live at different levels of the hierarchy. In particular, the finite-representation size $K$, the data-dependent occupancy count $c_n$, and the finite-mixture dimension $m$ are not directly comparable. We refer to Section~\ref{sec:counts} for a detailed discussion, including the relationship to mixtures of finite mixtures and the role of identifiability when interpreting $c_n$ under finite models.

Overall, the framework presented here strengthens the connection between random partitions and mixture modeling, offering a unified and tractable view of SSP-based priors and their computation. We anticipate that this perspective will facilitate further developments in Bayesian nonparametrics, particularly for more complex hierarchical and non-exchangeable structures where exact finite representations may offer both conceptual clarity and practical computational advantages.

\footnotesize
\bibliographystyle{apalike-ejor}
\bibliography{reference}

\end{document}